\documentclass[journal]{IEEEtran}

\ifCLASSINFOpdf
\else
   \usepackage[dvips]{graphicx}
\fi
\usepackage{url}

\usepackage{tikz}
\usepackage{amsmath,amsfonts,amssymb}%
\usepackage{bm}%
\usepackage{graphicx,graphics}%
\usepackage{cases}%
\usepackage[noadjust]{cite}%
\usepackage{color}%
\usepackage{cite,url}
\usepackage{verbatim}
\usepackage{balance}
\usepackage{multirow}
\usepackage{stfloats}
\usepackage{xspace}
\usepackage{amsthm}
\usepackage{mathtools}
\usepackage{subfig}
\usepackage{caption}
\usepackage{algorithm}
\usepackage{algpseudocode}
\usepackage{float}
\usepackage{ifpdf}

\floatname{algorithm}{Procedure}

\usepackage{tikz}
\usetikzlibrary{shapes.misc}
\usetikzlibrary{matrix}
\usetikzlibrary{arrows,backgrounds,fit,calc}

\graphicspath{{Figures/}}

\ifpdf
  \usepackage[pdftex,colorlinks=true,breaklinks=true,citecolor=blue]{hyperref}
\else
  \usepackage[ps2pdf,colorlinks=true,citecolor=blue]{hyperref}
\fi

\DeclarePairedDelimiterX\abs[1]{\lvert}{\rvert}{#1}
\DeclarePairedDelimiterX\parn[1]{(}{)}{#1}
\DeclarePairedDelimiterX\set[1]{\lbrace}{\rbrace}{#1}
\DeclarePairedDelimiterX\innerp[2]{\langle}{\rangle}{#1,#2}

\DeclarePairedDelimiterX\norm[1]{\lVert}{\rVert}{#1}
\DeclarePairedDelimiterX\brak[1]{\lbrace}{\rbrace}{#1}
\DeclarePairedDelimiterX\coeff[1]{(}{)}{#1}

\renewcommand{\hat}[1]{\widehat{#1}}

\newcommand{\untsph}{\mathbb{S}^{2}}
\newcommand{\lsph}{L^{2}(\untsph)}
\newcommand{\lSO}{L^{2}(\mathbb{SO}(3))}
\newcommand{\bv}[1]{\boldsymbol{#1}}
\newcommand{\unit}[1]{\bv{\hat{#1}}}

\newcommand{\dfn}{\triangleq}
\newcommand{\D}{\mathcal{D}}
\newcommand{\secref}[1]{Section\,\ref{#1}}
\newcommand{\figref}[1]{Fig.\,\ref{#1}}

\newcommand{\B}[1]{\mathbf{#1}}
\newcommand{\E}{\mathbb{E}}
\newcommand{\mse}{\mathcal{E}_{ms}}
\newcommand{\se}{\mathcal{E}_{s}}

\newcommand{\SO}{\mathbb{SO}(3)}
\newcommand{\snr}{\mathrm{SNR}}

\newcommand{\innerpS}[2]{\innerp{#1}{#2}_{\untsph}}
\newcommand{\innerpSO}[2]{\innerp{#1}{#2}_{\SO}}
\newcommand{\normS}[1]{\norm{#1}_{\untsph}}
\newcommand{\normSO}[1]{\norm{#1}_{\SO}}
\newcommand{\displayskipshrink}{%
	\setlength{\abovedisplayskip}{3.0pt plus 1.0pt minus 1.0pt}
	\setlength{\abovedisplayshortskip}{0pt plus 1.0pt minus 1.0pt}
	\setlength{\belowdisplayskip}{3.0pt plus 1.0pt minus 1.0pt}
	\setlength{\belowdisplayshortskip}{3.0pt plus 1.0pt minus 1.0pt}}

\newtheorem{theorem}{Theorem}

\title{Joint $\SO$-Spectral Domain Filtering of Spherical Signals in the Presence of Anisotropic Noise}

\author{Adeem Aslam, \IEEEmembership{Student Member, IEEE} and Zubair Khalid, \IEEEmembership{Senior Member, IEEE}
	\thanks{Authors are with the School of Science and Engineering, Lahore University of Management Science, Lahore, Pakistan (e-mail: adeem.aslam@lums.edu.pk, zubair.khalid@lums.edu.pk).
}}

\begin{document}

\displayskipshrink

\maketitle

\begin{abstract}
We present a joint $\SO$-spectral domain filtering framework using directional spatially localized spherical harmonic transform (DSLSHT), for the estimation and enhancement of random anisotropic signals on the sphere contaminated by random anisotropic noise. We design an optimal filter for filtering the DSLSHT representation of the noise-contaminated signal in the joint $\SO$-spectral domain. The filter is optimal in the sense that the filtered representation in the joint domain is the minimum mean square error estimate of the DSLSHT representation of the underlying (noise-free)~source signal. We also derive a least square solution for the estimate of the source signal from the filtered representation in the joint domain. 
We demonstrate the capability of the proposed filtering framework using the Earth topography map in the presence of anisotropic, zero-mean, uncorrelated Gaussian noise, and compare its performance with the joint spatial-spectral domain filtering framework.
\end{abstract}


\section{Introduction} \label{sec:intro}
Spherical signals have inherent angular dependence and are naturally encountered in many areas of science and engineering such as wireless communication~\cite{bashar:2016,Alem:2015}, computer graphics~\cite{nadeem2016spherical,Lam:2006}, medical imaging~\cite{Bates:2016}, acoustics~\cite{Zhang:2012,Bates:2015}, quantum chemistry~\cite{Johansson:2017}, quantum mechanics~\cite{grinter2018}, geodesy~\cite{Wieczorek:2007}, planetary sciences~\cite{Audet:2014,Hippel:2019}, astronomy~\cite{Starck:2006,Jarosik:2011,Galanti:2019} and cosmology~\cite{Spergel:2007,Dahlen:2008,Marinucci:2008,McEwen:2012,Nadathur:2016}, to name a few. Spherical observations in most of these fields are marred with unwanted, yet unavoidable, noise due to the presence of different sources of interference. In this context, we address the problem of estimating/recovering signals on the sphere which are contaminated by random anisotropic noise.

Many noise removal techniques, with different assumptions and constraints, have been proposed in the literature~\cite{McEwen:2008,McEwen:2012,sasgen2006wiener,arora2010optimal}. These methods process signals in either spatial or spectral domain and assume the signal and/or noise to be a realization of an isotropic random process on the sphere. To estimate signals contaminated by anisotropic noise, a minimum mean square error filter is developed in~\cite{klees2008design}, and zero-forcing and minimum mean square error criterion is adopted in~\cite{sadeghi2014} using linear operators for equalizing linear distortions and anisotropic noise. However, these methods do not carry out signal estimation in the joint spatial-spectral domain.

Motivated by the idea of filtering non-stationary processes in the joint time-frequency domain proposed in~\cite{hlawatsch:2000}, an optimal filter in the joint spatial-spectral domain has been proposed in \cite{Khalid:2013spie} for the estimation of spherical signals contaminated by zero-mean, anisotropic noise\footnote{Anisotropic processes on the sphere are analogues of non-stationary processes in the Euclidean domain.}. However, the resulting filter performs spatially varying filtering of the spectral content using axisymmetric window signals, and is therefore, not suitable for the recovery of directional features in the underlying signal.

The framework developed in this work is a novel contribution towards signal estimation on the sphere, aimed at recovering directional features, in the presence of random anisotropic noise. Before formulating the framework in \secref{sec:filtering}, we briefly review the mathematical background and formally state the problem in \secref{sec:Math}. In \secref{sec:ana}, we illustrate the utility of the proposed filtering framework on the bandlimited Earth topography map and compare the results with the joint spatial-spectral domain filtering framework~\cite{Khalid:2013spie}, before making concluding remarks in \secref{sec:con}.

\section{Preliminaries and Problem Formulation}\label{sec:Math}

\subsection{Signals on 2-Sphere}
Surface of the $2$-sphere~(or sphere), denoted by $\untsph$, is defined as $\untsph \dfn \{\unit{x} \in \mathbb{R}^3 \, : \, |\unit{x}| = 1\}$, where $\unit{x} \equiv \unit{x}(\theta,\phi) = (\sin\theta\cos\phi,\sin\theta\sin\phi,\cos\theta)^\mathrm{T}$ is parameterized by colatitude $\theta \in [0,\pi]$, measured from the positive $z$-axis, and longitude $\phi \in [0,2\pi)$, measured from the positive $x$-axis in the $x-y$ plane, $(\cdot)^\mathrm{T}$ represents the vector transpose and $|\cdot|$ denotes the Euclidean norm. Square-integrable and complex-valued functions, of the form $f(\unit{x}) \equiv f(\theta,\phi)$, defined over $\untsph$ form a Hilbert space, denoted by $\lsph$, that is equipped with the following inner product for two functions $f,h\in\lsph$
\begin{equation*}
\innerpS{f}{h} \dfn\! \int_{\untsph} f(\unit{x}) \overline{h(\unit{x})} ds(\unit{x}) = \!\int_{\theta=0}^{\pi} \int_{\phi=0}^{2\pi}\!\!f(\unit{x}) \overline{h(\unit{x})} ds(\unit{x})
\end{equation*}
where $ds(\unit{x}) = \sin\theta d\theta d\phi$ and $\overline{( \cdot )}$ represents the complex conjugate operation. This inner product induces a norm $\norm{f} \dfn \innerpS{f}{f}^{1/2}$. The Hilbert space $\lsph$ has a complete set of orthonormal basis functions called spherical harmonics, which are denoted by $Y_{\ell}^m(\unit{x})$ for integer degree $\ell \geq 0$ and integer order $-\ell \le m \leq \ell$~\cite{Kennedy-book:2013}. Any signal $f\in\lsph$ can then be expressed as
\begin{equation}
f(\unit{x}) = \sum\limits_{\ell = 0}^{\infty} \sum\limits_{m = -\ell}^{\ell} (f)_{\ell}^m Y_{\ell}^{m}(\unit{x})= \sum\limits_{n}^{\infty}  (f)_n Y_n(\unit{x}),
\label{eq:Fourier_expansion}
\end{equation}
where $\sum\limits_{n}^{\infty} \equiv \sum\limits_{\ell = 0}^{\infty} \sum\limits_{m = -\ell}^{\ell}$ with $n = \ell(\ell+1)+m$, and $(f)_n = \innerpS{f}{Y_n} = \innerpS{f}{Y_\ell^m}$ is the spectral coefficient of degree $\ell = \lfloor\sqrt{n}\rfloor$ and order $m = n - \ell(\ell+1)$, which constitutes the spectral domain representation of the  signal $f$. Signal $f \in \lsph$ is considered bandlimited to degree $L$ if $(f)_{\ell}^m = 0$ for $\ell, |m| \geq L$, or equivalently $(f)_n = 0$ for $n > L^2$. For such a signal, the sum over degree in~\eqref{eq:Fourier_expansion} is truncated at $L-1$.

\subsection{Signals on $\SO$ Rotation Group}
We define rotations on the sphere by three Euler angles, namely $\omega \in [0,2\pi)$ around $z$-axis, $\vartheta \in [0,\pi]$ around $y$-axis and $\varphi \in [0,2\pi)$ around $z$-axis, using the right handed $zyz$ convention~\cite{Kennedy-book:2013}. Group of all such proper rotations\footnote{An improper rotation is a reflection or a flip about either one of the axes or the center of the coordinate system.} is called the special orthogonal group, denoted by $\SO$, in which each point is represented by a $3$-tuple of the Euler angles as $\rho \equiv (\varphi,\vartheta,\omega)$. Square-integrable and complex-valued functions defined over the $\SO$ rotation group form a Hilbert space $\lSO$, which is equipped with the following inner product for two functions $g, \nu\in\lSO$
\begin{equation*}
\innerpSO{g}{\nu}\! \dfn\!\!\!\!\! \int\limits_{\SO}\!\!\!\!\! g(\rho) \overline{\nu(\rho)} d\rho \!= \!\!\!\!\int\limits_{\varphi=0}^{2\pi} \int\limits_{\vartheta=0}^{\pi} \int\limits_{\omega=0}^{2\pi} \!\!\!g(\rho) \overline{\nu(\rho)} d\varphi \sin\vartheta d\vartheta d\omega.
\end{equation*}
For the Hilbert space $\lSO$, Wigner-$D$ functions, denoted by $D^{\ell}_{m,m'}(\varphi,\vartheta,\omega)$ for integer degree $\ell \geq 0$ and orders $-\ell \le m,m' \leq \ell$, form a complete set of orthogonal basis functions~\cite{Kennedy-book:2013}. Any signal $g\!\in\!\lSO$ can be expanded as
\begin{align}
g(\rho) = \sum_{\ell=0}^{\infty} \sum_{m=-\ell}^{\ell} \sum_{m'=-\ell}^{\ell} (g)^{\ell}_{m,m'} D^{\ell}_{m,m'}(\rho), 
\label{eq:WignerD_expansion}
\end{align}
where
$(g)^{\ell}_{m,m'} = \left(\frac{2\ell+1}{8\pi^2}\right)\innerpSO{g}{D^{\ell}_{m,m'}}$ is the spectral coefficient of degree $\ell \ge 0$ and orders $-\ell \le m, m' \le \ell$, and constitutes the spectral domain representation of the signal $g$. Signal $g \in \lSO$ is considered bandlimited to degree $L$ if $(g)^{\ell}_{m,m'} = 0$ for $\ell,|m|,|m'| \ge L$.

\subsection{Problem Under Consideration}
\label{sec:PS}
Let $s(\unit{x})$ be a realization of an anisotropic random process on the sphere, called the source signal, which is contaminated by a realization of an anisotropic, zero-mean, random noise process, $z(\unit{x})$, to give $f(\unit{x}) = s(\unit{x}) + z(\unit{x})$ as an observation on the sphere. The objective is to determine an estimate of the source signal, denoted by $\tilde{s}(\unit{x})$, which is optimal in the mean square sense. We assume that noise is uncorrelated with the source signal\footnote{We refer to this assumption by simply stating ``uncorrelated noise".}, i.e., $\E\{s(\unit{x}) \overline{z(\unit{x})}\} = \E\{z(\unit{x})\overline{s(\unit{x})}\} = \E\{(s)_n \overline{(z)_{n'}}\} = \E\{(z)_n \overline{(s)_{n'}}\} = 0$, where $\E\{\cdot\}$ denotes the expectation operator, and spectral covariance matrices for the signal and noise processes, denoted by $\B{C}^s$ and $\B{C}^z$ respectively, are known.

\section{Joint $\SO$-Spectral Domain Filtering}
\label{sec:filtering}
Before presenting the joint $\SO$-spectral domain filtering framework, we introduce the joint $\SO$-spectral domain representation using directional spatially localized spherical harmonic transform (DSLSHT), which, for a signal $f\in\lsph$, bandlimited to degree $L_f$, is defined as~\cite{Khalid:2013DSLSHT}
\begin{equation}
g_f(\rho;u) = \int_{\untsph} f(\unit{x}) (\D(\rho)h)(\unit{x}) \overline{Y_u(\unit{x})} ds(\unit{x}),
\label{eq:DSLSHT}
\end{equation}
where $u = v(v+1)+w$ for $0 \le v \le L_f-1$, $-v \le w \le v$, $h \in \lsph$ is the window signal  bandlimited to degree $L_h$ and is required to be spatially concentrated with in some region on the sphere to provide spatial localization for the signal $f$, $\D(\rho)$ is the rotation operator and $g_f(\rho;u)$ is called the DSLSHT representation of the signal $f$. Using \eqref{eq:Fourier_expansion} and the spectral representation of the rotated signal, given by $\innerpS{\D h}{Y_{\ell}^m} = \sum\limits_{m'=-\ell}^{\ell} D^{\ell}_{m,m'}(\varphi,\vartheta,\omega)(h)_{\ell}^{m'}$~\cite{Kennedy-book:2013}, we can rewrite the DSLSHT representation in \eqref{eq:DSLSHT} as
\begin{align}
g_f(\rho;u) &= \sum_{n=0}^{N_f} (f)_n \psi_{u,n}(\rho),
\label{eq:DSLSHT2}
\end{align}
where $N_f = L^2_f-1$,
\begin{align}
\psi_{u,n}(\rho) \dfn \sum_{p = 0}^{L_h-1}\sum_{q=-p}^p\sum_{q'=-p}^p D^p_{q,q'}(\rho) (h)_p^{q'} T(n;p,q;u),
\label{eq:psi}
\end{align}
and $T(n;p,q;u) = \displaystyle\int_{\untsph} Y_n(\unit{x}) Y_p^q(\unit{x}) \overline{Y_u(\unit{x})} ds(\unit{x})$ is the spherical harmonic triple product~\cite{Kennedy-book:2013}. From \eqref{eq:DSLSHT2}, \eqref{eq:psi} and the definition of spherical harmonic triple product, we note that the bandlimit of $g_f(\rho;u)$ in $\rho$ and $u$ is given by $L_h$ and $L_g = L_f+L_h-1$ respectively.

\subsection{Joint $\SO$-Spectral Domain Filter Design}
\label{sec:filter}
We define the joint $\SO$-spectral domain filter function as
\begin{align}
\zeta(\rho;u) = \sum_{p=0}^{L_{\zeta_u}-1} \sum_{q=-p}^p \sum_{q'=-p}^p \big(\zeta(\cdot;u)\big)^{p}_{q,q'}\,\, D^{p}_{q,q'}(\rho),
\label{eq:zeta_comp}
\end{align}
for $u=0,1,\ldots,N_g\!= \!L^2_g-1$. Filtering in the joint domain is carried out by convolving the DSLSHT representation of the signal, $g_f(\rho;u)$, with the joint $\SO$-spectral domain filter function $\zeta(\rho;u)$ for each spectral component $u$, i.e.,
\begin{align}
\nu(\rho;u) \!=\! \sum_{p=0}^{L_h-1} \sum_{q,q'=-p}^p \sum_{k=-p}^p \big(g_f(\cdot;u)\big)^p_{k,q'} \big(\zeta(\cdot;u)\big)^p_{q,k}  D^p_{q,q'}(\rho),
\label{eq:filtered_distribution}
\end{align}
where $\nu(\rho;u)$ is called the filtered representation and we have used the definition of convolution of $\SO$ signals given in~\cite{Kostelec:2008}. Moreover, we have used the fact that bandlimit of $g_f(\rho;u)$ in $\rho$ is $L_h$ and have assumed, without loss of generality, that each filter component, i.e., $\zeta(\rho;\cdot)$, is also bandlimited in $\rho$ to $L_h$. Filter function in \eqref{eq:zeta_comp} is obtained by minimizing the following mean square error in the joint $\SO$-spectral domain
\begin{align}
\mse = \E \left\{ \sum_{u=0}^{N_g} \normSO{\nu(\rho;u) - g_s(\rho;u)}^2 \right\},
\label{eq:mse}
\end{align}
where $N_g = L_g^2-1$ and $g_s(\rho;u)$ is the source signal DSLSHT representation. We present the results in the following theorem.
\begin{theorem}
	\label{th:theorem1}
	Let $f(\unit{x}) = s(\unit{x}) + z(\unit{x})$ be a noise-contaminated random observation on the sphere, where $s(\unit{x})$ is a realization of an anisotropic random process of interest, called the source signal, and $z(\unit{x})$ is a realization of an anisotropic, zero-mean random process, representing the noise signal. Assuming that the source and noise signals are uncorrelated with known spectral covariance matrices, denoted by $\B{C}^s$ and $\B{C}^z$ respectively, the joint $\SO$-spectral domain filter, which minimizes the mean square error defined in \eqref{eq:mse}, is obtained by inverting the following linear system	
	\begin{align}
		\B{A}(p,u) \B{F}(p,q,u) = \B{b}(p,q,u),
		\label{eq:zeta_final}
	\end{align}
	for $0 \le p \le L_h-1, \, |q| \le p, \, u \le N_g = L_g^2-1$, where $\B{F}(p,q,u)$ is a column vector of size $(2p+1)$, with elements given by $F_k = \big(\zeta(\cdot;u)\big)^p_{q,k}, \, |k| \le p$. Elements of the matrix $\B{A}$ and column vector $\B{b}$ are given by
	\begin{align}
		A_{k',k} = \sum\limits_{n=0}^{N_f}\sum\limits_{n'=0}^{N_f} T(n;p,k;u) \overline{T(n';p,k';u)} (C^s_{nn'} + C^z_{nn'}),
		\label{eq:A_elements}
	\end{align}
	\begin{align}
		b_{k'} = \sum\limits_{n=0}^{N_f} \sum\limits_{n'=0}^{N_f} T(n;p,q;u) \overline{T(n';p,k';u)} C^s_{nn'},
		\label{eq:b_elements}
	\end{align}
	respectively, for $|k|, |k'| \le p$. $C^s_{nn'} = \E\{(s)_n\overline{(s)_{n'}}\}$ and $C^z_{nn'} = \E\{(z)_n\overline{(z)_{n'}}\}$ are the elements of $\B{C}^s$ and $\B{C}^z$ respectively.
\end{theorem}
\begin{proof}
	Using \eqref{eq:WignerD_expansion}, the expression for filtered representation in \eqref{eq:filtered_distribution} and the orthogonality of Wigner-$D$ functions on the $\SO$ rotation group~\cite{Kennedy-book:2013}, mean square error in \eqref{eq:mse} can be written as
	\begin{align}
	&\mse = \sum_{u=0}^{N_g} \sum_{p=0}^{L_h-1} \sum_{q,q'=-p}^p \left(\frac{8\pi^2}{2p+1}\right) \times \nonumber \\
	& \E\Bigg\{ \Bigg(\sum_{k=-p}^p \big(g_f(\cdot;u)\big)^p_{k,q'} \big(\zeta(\cdot;u)\big)^p_{q,k} - \big(g_s(\cdot;u)\big)^p_{q,q'}\Bigg) \times \nonumber \\
	& \overline{\Bigg(\sum_{k'=-p}^p \big(g_f(\cdot;u)\big)^p_{k',q'} \big(\zeta(\cdot;u)\big)^p_{q,k'} - \big(g_s(\cdot;u)\big)^p_{q,q'} \Bigg)} \Bigg\},
	\label{eq:mse2}
	\end{align}
	where for a signal $d\in\lsph$ with bandlimit $L_d$ such that $N_d = L^2_d-1$,
	\begin{align}
	\big(g_d(\cdot;u)\big)^p_{q,q'} = \sum_{n=0}^{N_d}(d)_n (h)_p^{q'} T(n;p,q;u).
	\label{eq:G}
	\end{align}
	Substituting \eqref{eq:G} in \eqref{eq:mse2}, setting the derivative of the resulting expression for $\mse$ with respect to $\overline{\big(\zeta(\cdot;u)\big)^p_{q,k'}}$ equal to zero and noting the fact that signal and noise are uncorrelated, we obtain a linear system which, using \eqref{eq:A_elements} and \eqref{eq:b_elements}, can be cast in the matrix form given in \eqref{eq:zeta_final}\footnote{We note that the linear system in \eqref{eq:zeta_final} becomes ill-conditioned for certain values of $p,q$ and $u$, in which case we use the Moore-Penrose pseudo-inverse to obtain the filter coefficients $\B{F}(p,q,u)$.}.
\end{proof}

\subsection{Signal Estimation}
\label{sec:sig_est}
The filtered representation $\nu(\rho;u)$ may not be an admissible DSLSHT representation, i.e., there may not exist a signal $\tilde{s} \in \lsph$ such that $g_{\tilde{s}}(\rho;u) = \nu(\rho;u)$. As a result, we cannot use the inverse DSLSHT~\cite{Khalid:2013DSLSHT} to obtain the source signal estimate from $\nu(\rho;u)$. We present a least square estimate of the source signal in the following theorem.
\begin{theorem}
	\label{th:theorem2}
	Let $g_f(\rho;u)$ be the DSLSHT representation of the noise-contaminated random signal on the sphere which is filtered in the joint $\SO$-spectral domain using the filter coefficients obtained from \eqref{eq:zeta_final}, resulting in a filtered representation $\nu(\rho;u)$ given in \eqref{eq:filtered_distribution}. By minimizing the following squared error
	\begin{align}
	\se = \sum\limits_{u=0}^{N_g} \, \normSO{\nu(\rho;u) - g_{\tilde{s}}(\rho;u)}^2,	\quad N_g = L^2_g-1,
	\label{eq:se}
	\end{align}
	an estimate of the source signal, denoted by $\tilde{s}(\unit{x})$, can be obtained by solving the following linear system
	\begin{align}
	\B{\tilde{s}} = \bv{\Upsilon} \B{f},
	\label{eq:opt_prob_sol}
	\end{align}
	where $\B{\tilde{s}}$ and $\B{f}$ are the column vectors containing spectral coefficients of the estimate $(\tilde{s})_n$ and the noise-contaminated observation $(f)_n$ respectively for $0 \le n \le N_f=L^2_f-1$, and $\bv{\Upsilon}$ is an $L_f^2\times L_f^2$ matrix with elements given by
	\begin{align}
		\Upsilon_{n,n'} &= \frac{4\pi}{\innerpS{h}{h}} \sum\limits_{u=0}^{N_g} \sum\limits_{p=0}^{L_h-1} \frac{1}{(2p+1)} \sum\limits_{q=-p}^p \left(\sum\limits_{q'=-p}^p |(h)_p^{q'}|^2\right) \times \nonumber \\
		& \quad \sum_{k=-p}^p \big(\zeta(\cdot;u)\big)^p_{q,k} \overline{T(n;p,q;u)} T(n';p,k;u).
		\label{eq:upsilon}
	\end{align}
\end{theorem}
\begin{proof}
	Rewriting the squared error in \eqref{eq:se} using \eqref{eq:DSLSHT2} and setting the derivative of the resulting expression with respect to $(\tilde{s})_n$ equal to zero, we get the following relation
	\begin{align}
	&\sum_{n'=0}^{N_f} \left( \sum_{u=0}^{N_g} \displaystyle \int_{\SO} \psi_{u,n'}(\rho) \overline{\psi_{u,n}(\rho)} d\rho \right) (\tilde{s})_{n'}  = \nonumber \\
	& \qquad \sum_{u=0}^{N_g} \displaystyle \int_{\SO} \nu(\rho;u) \overline{\psi_{u,n}(\rho)} d\rho, \quad  0 \le n \le N_f.
	\label{eq:opt_prob_step1}
	\end{align}
	From the definition of $\psi_{u,n}(\rho)$ in \eqref{eq:psi}, we can write
	\begin{align}
	& \sum_{u=0}^{N_g} \displaystyle \int_{\SO} \psi_{u,n'}(\rho) \overline{\psi_{u,n}(\rho)} d\rho = \sum_{u=0}^{N_g} \sum_{p = 0}^{L_h-1} \frac{8\pi^2}{(2p+1)} \times \nonumber \\
	& \qquad \qquad  \sum_{q=-p}^p\sum_{q'=-p}^p |(h)_p^{q'}|^2 T(n';p,q;u) T(n;p,q;u),
	\label{eq:opt_prob_step2}
	\end{align}
	where we have used the orthogonality of Wigner-$D$ functions on the $\SO$ rotation group. Using the conjugate symmetry of spherical harmonics\footnote{Conjugate symmetry: $\overline{Y_{\ell}^m(\unit{x})} = (-1)^m Y_{\ell}^{-m}(\unit{x})$.}, the definition of spherical harmonic triple product in terms of Wigner-$3j$ symbols~\cite{Kennedy-book:2013}, which are represented by $\left(
	\begin{matrix}
	\ell & p & v \\
	m & q & w \\
	\end{matrix}
	\right)$, and the following relations
	\begin{align*}
	\sum_{w=-v}^v \sum_{q=-p}^p \left(
	\begin{matrix}
	\ell & p & v \\
	m & q & -w \\
	\end{matrix}
	\right)
	\left(
	\begin{matrix}
	\ell' & p & v \\
	m' & q & -w \\
	\end{matrix}
	\right) = \frac{\delta_{\ell,\ell'}\delta_{m,m'}}{(2\ell+1)},
	\end{align*}
	\begin{align*}
	\sum_{v=0}^{L_g-1} (2v+1) \left(
	\begin{matrix}
	\ell & p & v \\
	0 & 0 & 0 \\
	\end{matrix}
	\right)^2 = 1,
	\end{align*}
	the right hand side of \eqref{eq:opt_prob_step2} can be simplified as $2\pi \innerpS{h}{h}  \delta_{n,n'}$. Hence, \eqref{eq:opt_prob_step1} gives the spectral estimate as
	\begin{align}
	(\tilde{s})_n &= \left(2\pi \innerpS{h}{h} \right)^{-1} \sum\limits_{u=0}^{N_g} \displaystyle \int_{\SO} \nu(\rho;u) \overline{\psi_{u,n}(\rho)} d\rho.
	\label{eq:opt_prob_step5}
	\end{align}
	Simplifying the integral in \eqref{eq:opt_prob_step5} using \eqref{eq:psi} and \eqref{eq:filtered_distribution}, and combining the result with \eqref{eq:G}, we formulate the source signal spectral estimate in the matrix form given in \eqref{eq:opt_prob_sol}.
\end{proof}

\begin{figure*}[!th]
	\vspace{-6mm}
	\centering
	\subfloat[$s(\unit{x})$]{
		\includegraphics[width=0.115\textwidth]{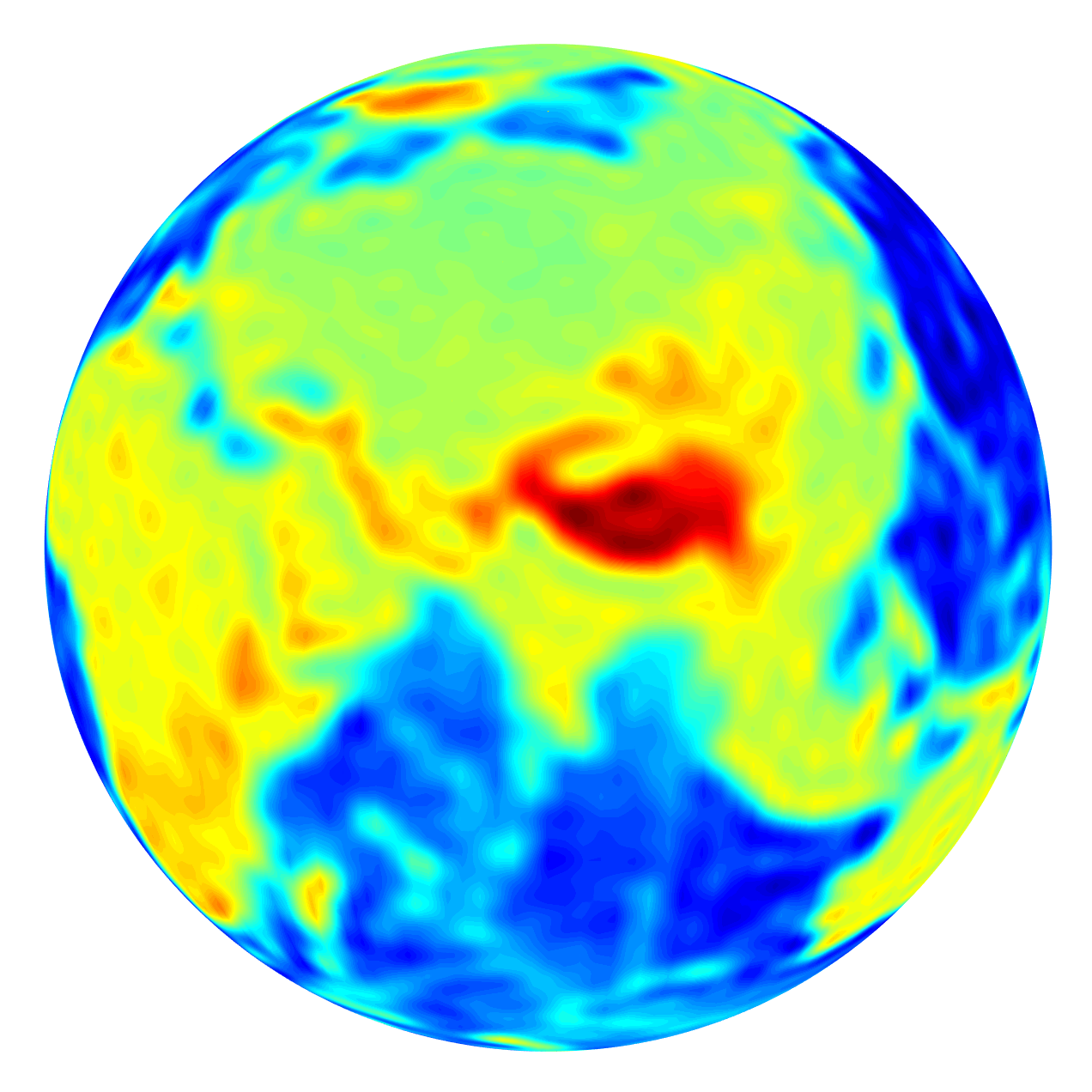}}\hfil
	\subfloat[$z(\unit{x})$]{
		\includegraphics[width=0.115\textwidth]{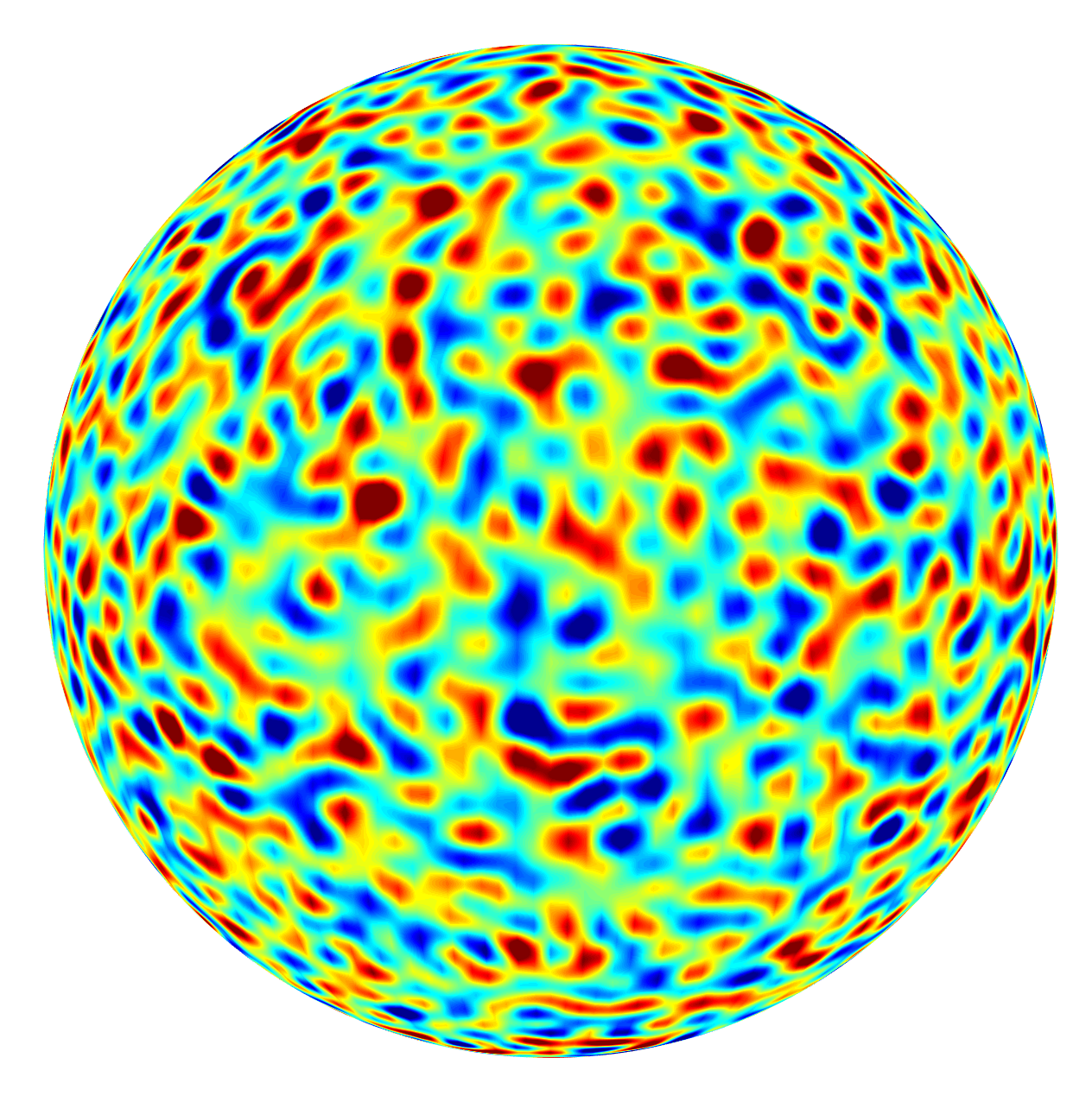}}\hfil	
	\subfloat[$f(\unit{x})$]{
		\includegraphics[width=0.115\textwidth]{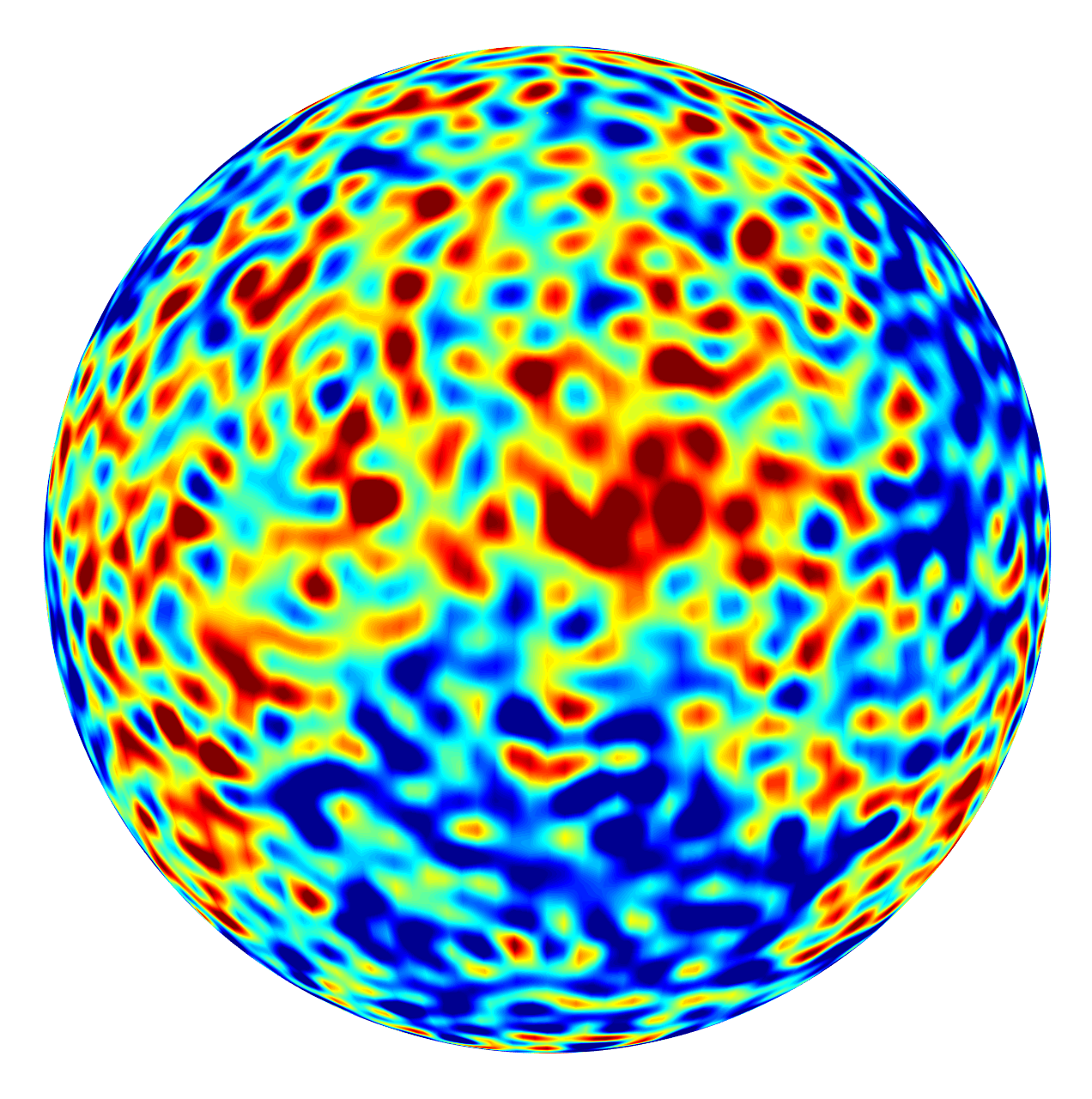}}\hfil
	\subfloat[$\tilde{s}(\unit{x})$]{
		\includegraphics[width=0.115\textwidth]{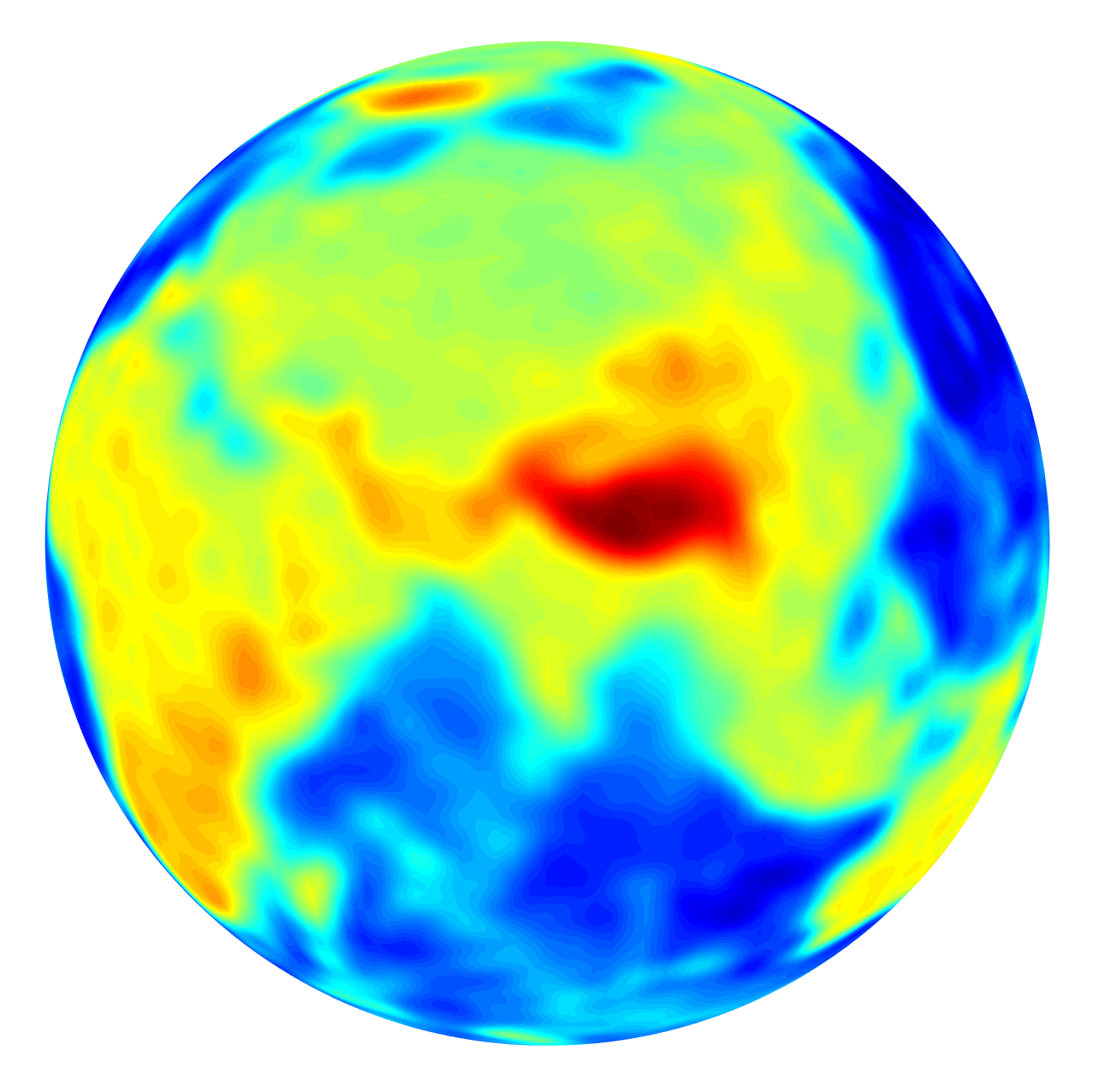}}\hfil
		\subfloat[$\tilde{s}_0(\unit{x})$]{
		\includegraphics[width=0.115\textwidth]{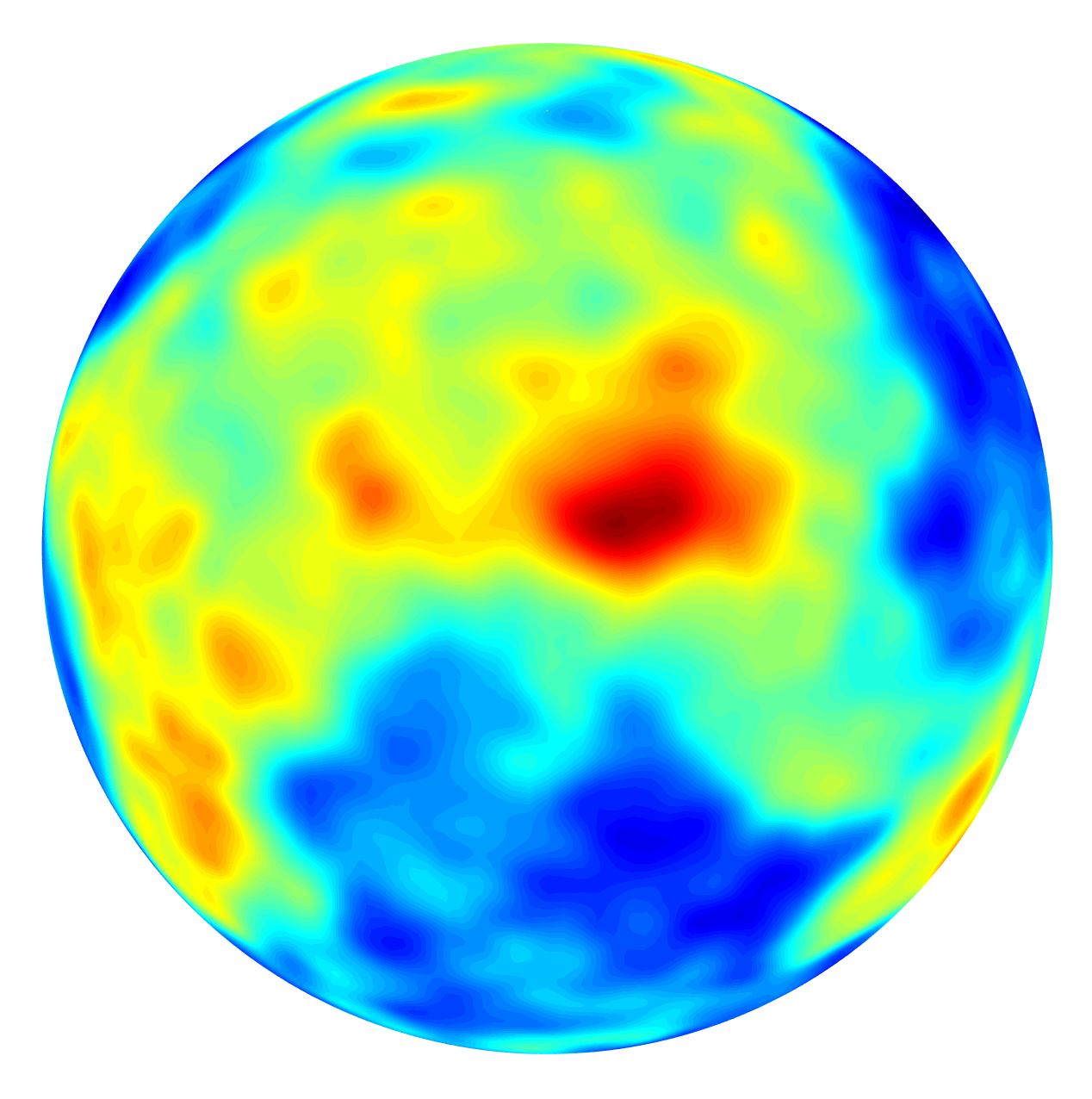}}\hfil
	\subfloat[$|s(\unit{x})|$]{
		\includegraphics[width=0.115\textwidth]{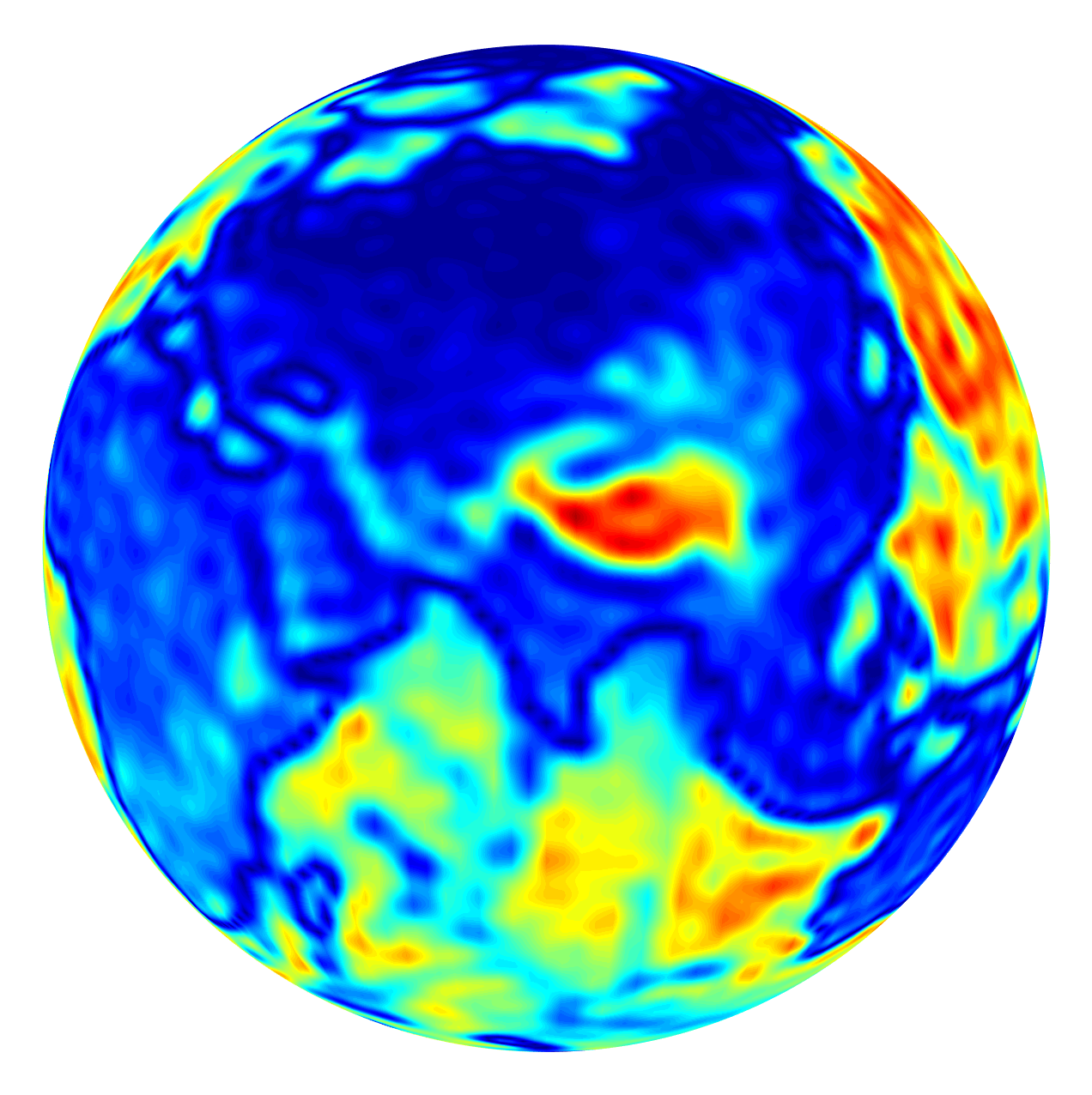}}\hfil
		\subfloat[$|\tilde{s}(\unit{x})|$]{
		\includegraphics[width=0.115\textwidth]{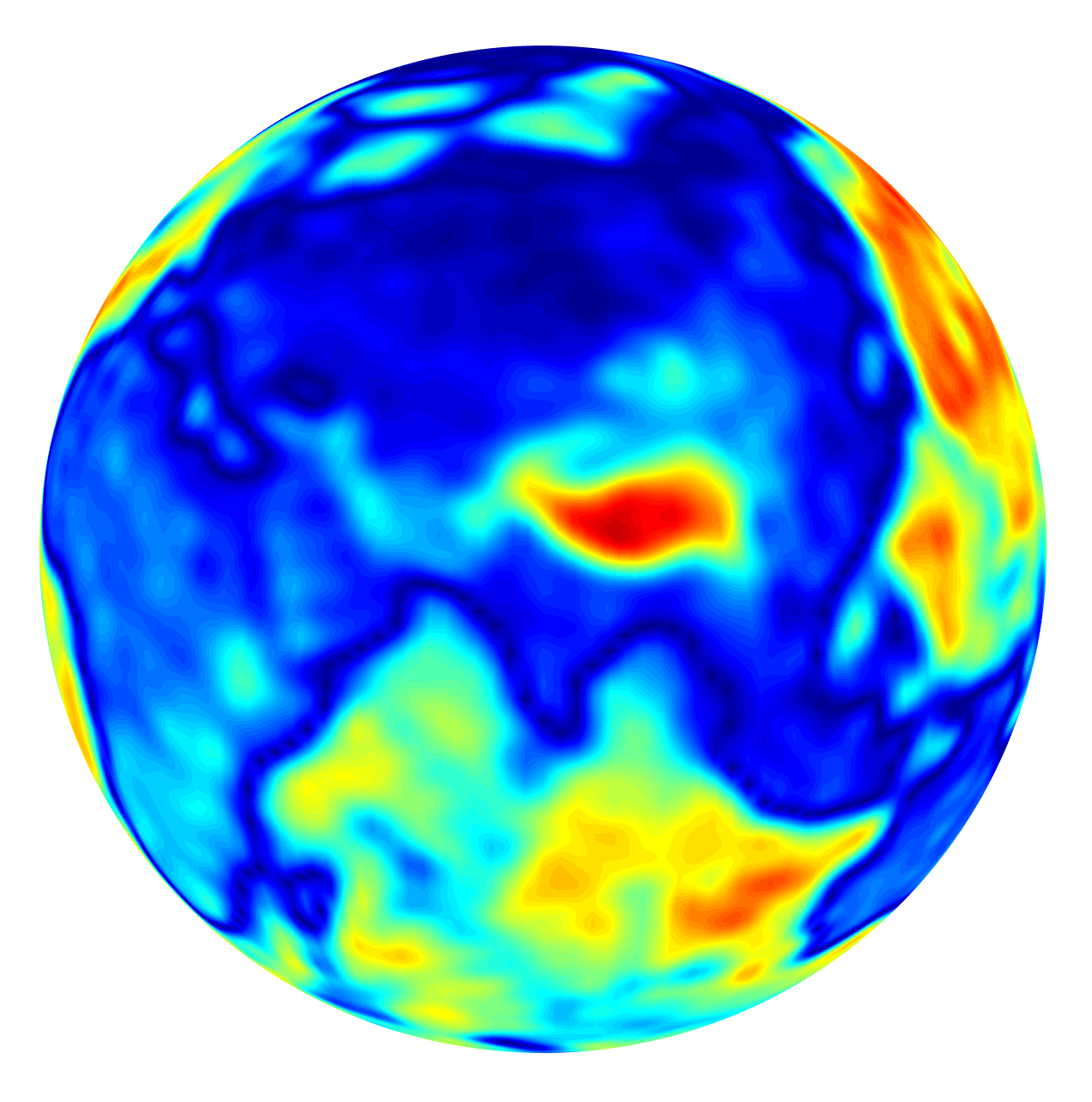}}\hfil
		\subfloat[$|\tilde{s}_0(\unit{x})|$]{
		\includegraphics[width=0.115\textwidth]{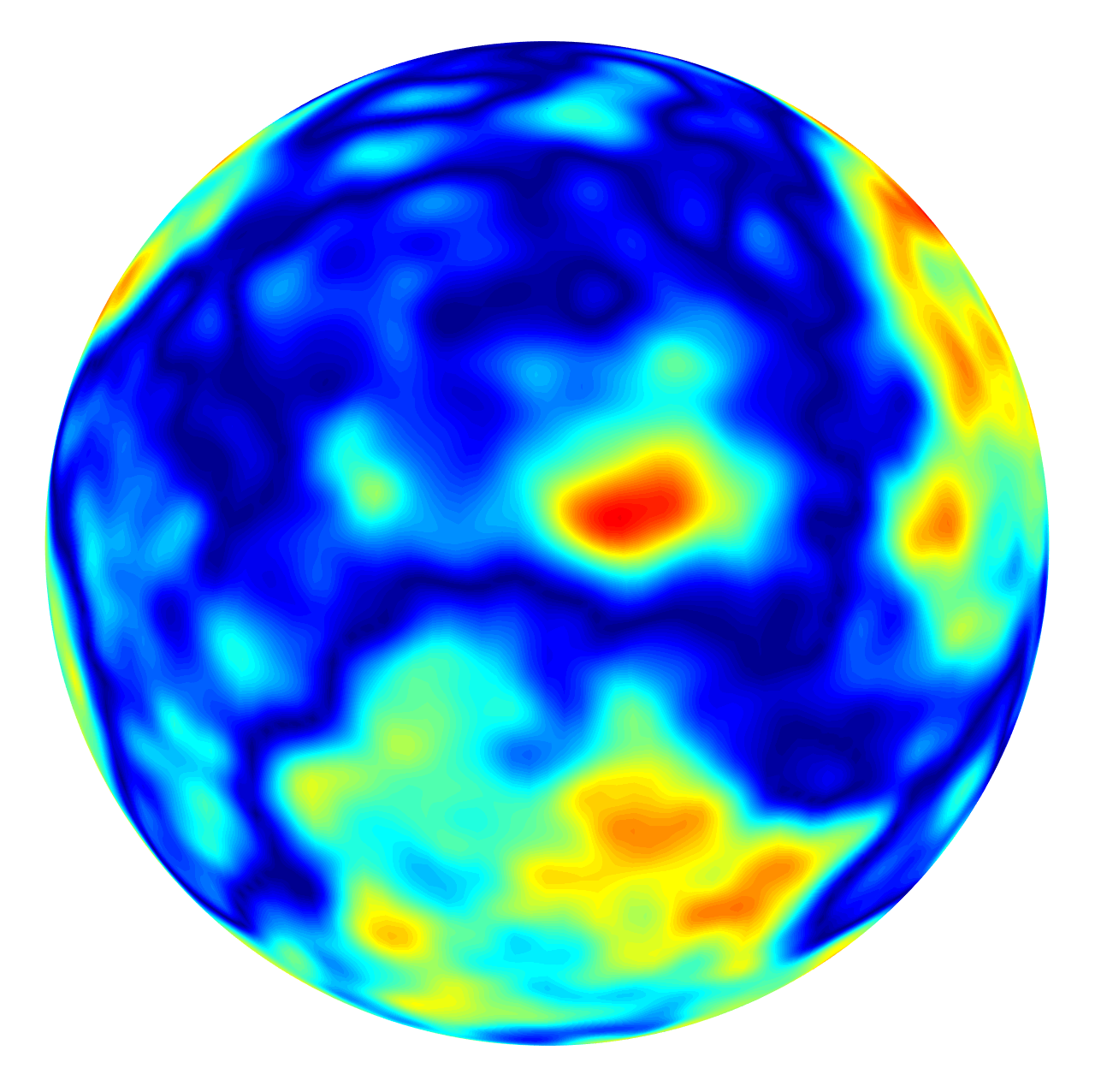}}\hfil
	\vspace{-2mm}
		
	\subfloat{
		\includegraphics[width=0.62\textwidth]{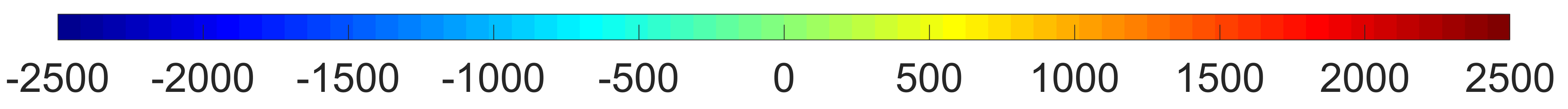}}\hfill
	\subfloat{
		\includegraphics[width=0.36\textwidth]{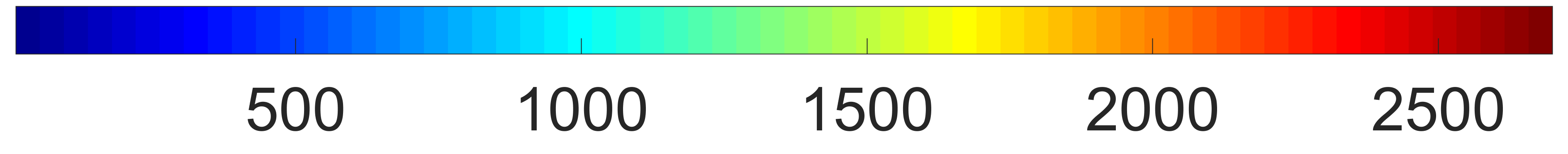}}\hfill
	\caption{\small Joint $\SO$-spectral domain filtering of the Earth topography map~($L_f\!=\!64$), $s(\unit{x})$, contaminated with anisotropic, zero-mean, uncorrelated Gaussian noise, $z(\unit{x})$, at $\snr^f \!= 0.001\!$ dBs. Also shown is the estimate obtained from the joint spatial-spectral domain filtering,  $\tilde{s}_0(\unit{x})$. Joint $\SO$-spectral domain filtering outperforms the joint spatial-spectral domain filtering by $8$ dBs. Moreover, magnitude plots show much better reconstruction of the directional features, such as the dark blue contours marking the boundary between land and water, using the proposed framework.}
	\label{fig:results}
	\vspace{-2mm}
\end{figure*}

\section{Analysis}
\label{sec:ana}
To demonstrate the effectiveness of the joint $\SO$-spectral domain filtering framework, we use the Earth topography map\footnote{http://geoweb.princeton.edu/people/simons/software.html}, bandlimited to $L_f = 64$, as the source signal $s(\unit{x})$ and gauge the performance using the signal to noise ratio~(SNR) defined as
\begin{align}
\snr^d = 20\log\frac{\normS{s(\unit{x})}}{\normS{d(\unit{x}) - s(\unit{x})}}
\label{eq:SNR}
\end{align}
for a signal $d \in \lsph$. Hence, the input and output SNRs are given by $\snr^f$ and $\snr^{\tilde{s}}$ respectively. Spectral covariance matrices for the source and noise signals are constructed as $\B{C}^s = \B{s}\B{s}^{\mathrm{H}}$ and $\B{C}^z = \B{T} \B{T}^{\mathrm{H}}$ respectively, where $\B{s}$ is the column vector containing the spectral coefficients $(s)_n$, elements of the matrix $\B{T}$ are chosen to be uniformly distributed in the interval $(-1,1)$ in both real and imaginary parts, and $(\cdot)^{\mathrm{H}}$ represents the conjugate transpose. We employ the most optimally concentrated Slepian function~\cite{Simons:2006}, computed for an elliptical region of focus colatitude $\theta_C \!=\! 15^{\circ}$ which is centered at the north pole such that the semi-major axis is aligned with the $x$-axis and has a radius $a \!=\! 16^{\circ}$\footnote{We refer the reader to~\cite{Khalid:2013DSLSHT} for the definition of a spherical ellipse.}, as the window signal $h$, with bandlimit $L_h \!=\! 20$, for computing the source signal estimate $\tilde{s}(\unit{x})$. For comparison, we filter the noise-contaminated observation $f(\unit{x})$ using the joint spatial-spectral domain filter in~\cite{Khalid:2013spie}, employing the most optimally concentrated Slepian function, computed for the axisymmetric polar cap region with polar cap angle $\theta_0 = 15^{\circ}$, as the window signal $h_0$~($L_{h_0}=20$), and obtain the estimate as $\tilde{s}_0(\unit{x})$.
\begin{figure}[!t]
	\vspace{-4mm}
	\centering
	\includegraphics[width=0.45\textwidth]{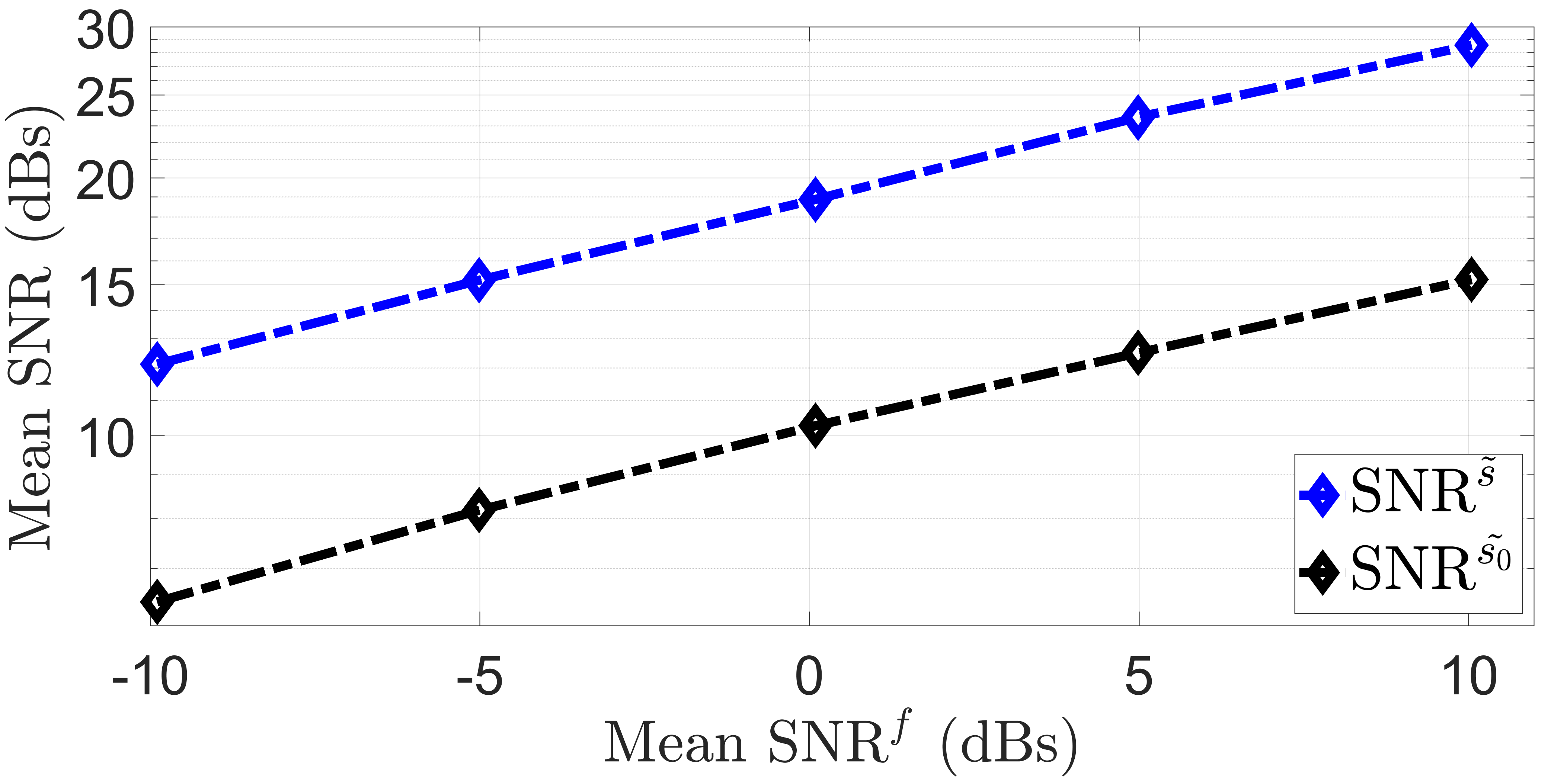}
	\vspace{-1mm}
	\caption{\small Mean output SNR plotted against the mean input SNR for $10$ realizations of the anisotropic, zero-mean, uncorrelated Gaussian noise process. Blue and black curves show the results using the proposed joint $\SO$-spectral domain and the joint spatial-spectral domain filtering frameworks respectively.}
	\vspace{-4mm}
	\label{fig:avg_SNR_trend}	
\end{figure}

As an illustration, we use a realization of an anisotropic, zero-mean, uncorrelated Gaussian noise process, $z(\unit{x})$, to obtain the noise-contaminated observation $f(\unit{x}) = s(\unit{x}) + z(\unit{x})$ such that $\snr^f \! = 0.001\!$ dBs. Output SNR using the proposed framework, i.e., $\snr^{\tilde{s}}$, is measured to be $18.33$ dBs, indicating a significant gain in SNR, compared to the joint spatial-spectral domain filtered estimate which results in $\snr^{\tilde{s}_0} = 10.36$ dBs. As expected, the joint $\SO$-spectral domain filtering framework outperforms the joint spatial-spectral domain filtering framework~(by $8$ dBs) due to its ability to better detect the underlying directional features of the data. The results are shown in \figref{fig:results}, where in addition to better reconstruction using the proposed framework, better estimate of the directional features of the Earth topography map can be observed, e.g., in the dark blue contours marking the boundary between land and water, as depicted in the magnitude plots.

To test the robustness of the proposed framework, we contaminate the Earth topography map with $10$ realizations of anisotropic, zero-mean, uncorrelated Gaussian noise process. We conduct a similar experiment for the joint spatial-spectral domain filtering framework. The results are averaged over realizations and plotted in \figref{fig:avg_SNR_trend} which shows the mean output SNR against the mean input SNR. As can be seen, the joint $\SO$-spectral domain filter performs better, even at severely high noise levels, compared to the joint spatial-spectral domain filter.

\section{Conclusion}
\label{sec:con}
We have presented a framework for the joint $\SO$-spectral domain filtering and estimation of random anisotropic signals contaminated by random anisotropic noise, using the directional spatially localized spherical harmonic transform (DSLSHT). We have designed a filter, which is optimal in the sense of mean square error criterion in the joint $\SO$-spectral domain, for filtering the DSLSHT representation of the noise-contaminated signal, and have proposed a least square solution for the estimate of the underlying (noise-free)~source signal from the filtered representation. We have demonstrated the capability of the proposed filtering framework on the bandlimited Earth topography map in the presence of anisotropic, zero-mean, uncorrelated Gaussian noise, and have shown that the proposed filtering framework performs much better compared to the joint spatial-spectral domain filter.

\bibliography{IEEEabrv,sht_bib}

\end{document}